%% file: main.tex
\documentclass[a4paper,UKenglish]{article}
\usepackage{latexsym}
\usepackage{amsmath,amssymb,mathtools}
\usepackage{cite}
\usepackage{comment}
\usepackage{apxproof} 

\usepackage{authblk}
\usepackage{amsthm}           

\usepackage[linesnumbered,lined,ruled,noend]{algorithm2e}

\input{macro}

\title{Efficient Enumeration of Dominating Sets for Sparse Graphs}
\author[1]{Kazuhiro Kurita}
\author[2]{Kunihiro Wasa}
\author[1]{Hiroki Arimura}
\author[2]{Takeaki Uno}

\affil[1]{IST, Hokkaido University, Sapporo, Japan\\
  \texttt{\{k-kurita, arim\}@ist.hokudai.ac.jp}}
\affil[2]{National Institute of Informatics, Tokyo, Japan\\
  \texttt{\{wasa, uno\}@nii.ac.jp}}

\begin{document}
\maketitle
\begin{abstract}
    \input{abst}
\end{abstract}
\input{intro}
\input{basic}
\input{degenerate}

\input{girth}
\input{conc}

\bibliographystyle{abbrv}
\bibliography{main.bbl}

\end{document}

%% file: macro.tex
\newcommand{\set}[1]{\left\{#1\right\}}
\newcommand{\relmiddle}[1]{\mathrel{}\middle#1\mathrel{}}
\newcommand{\inset}[2]{\left\{#1 \relmiddle| #2\right\}}
\newcommand{\name}[1]{\emph{#1}}

\newcommand{\size}[1]{\left| #1 \right|}
\newcommand{\order}[1]{O\left(#1\right)}
\newcommand{\bigomega}[1]{\Omega\left(#1\right)}

\newtheorem{lemma}{Lemma}
\newtheorem{theorem}[lemma]{Theorem}
\newtheorem{problem}{Problem}

\newcommand{\cand}[1]{C\left(#1\right)}

\newcommand{\sig}[1]{\mathcal{#1}}

\newcommand{\Del}[3][]{Del_{#1}\left(#2, #3\right)}
\newcommand{\DelG}[3][]{Del'_{#1}\left(#2, #3\right)}

\newcommand{\Fam}[1]{\mathcal{F}\left(#1\right)}
\newcommand{\Sol}[1]{\mathcal{S}\left(#1\right)}
\newcommand{\PC}[1]{\mathcal{E}\left(#1\right)}
\newcommand{\Parent}[1]{\mathcal{P}\left(#1\right)}
\newcommand{\pv}[1]{pv\left(#1\right)}

\newcommand{\True}{\mathtt{True}}
\newcommand{\False}{\mathtt{False}}

\makeatletter
\@ifpackageloaded{algorithm2e}{
\SetKwProg{Fn}{Function}{}{}
\SetKwProg{Procedure}{Procedure}{}{}
\SetKwProg{Subprocedure}{Subprocedure}{}{}
\SetKwComment{tcc}{//}{}
\SetKwFunction{Output}{Output}%
\SetKwFunction{AllChildren}{AllChildren}
\SetKwFunction{CandD}{Cand-D}
\SetKwFunction{CandG}{Cand-G}
\SetKwFunction{DomList}{DomList}
\SetKwFunction{Update}{Update}
\SetKwInOut{AlgInput}{Input}
\SetKwInOut{AlgOutput}{Output}
\SetKwInOut{AlgPrecondition}{Pre-conditions}
\SetKwInOut{AlgInvariant}{Invariants}

\SetKwFunction{EnumDS}{EDS}
\SetKwFunction{EnumDSD}{EDS-D}
\SetKwFunction{EnumDSG}{EDS-G}
}{}
\makeatother

%% file: abst.tex
A dominating set $D$ of a graph $G$ is a set of vertices such that 
any vertex in $G$ is in $D$ or its neighbor is in $D$. 
Enumeration of minimal dominating sets in a graph is one of central problems in enumeration study 
since enumeration of minimal dominating sets corresponds to enumeration of minimal hypergraph transversal. 
However, enumeration of dominating sets including non-minimal ones has not been received much attention.
In this paper, 
we address enumeration problems for dominating sets from sparse graphs which are degenerate graphs and graphs with large girth, 
and we propose two algorithms for solving the problems. 
The first algorithm enumerates all the dominating sets for a $k$-degenerate graph in $\order{k}$ time per solution using $\order{n + m}$ space, 
where $n$ and $m$ are respectively the number of vertices and edges in an input graph. 
That is, the algorithm is optimal for graphs with constant degeneracy such as trees, planar graphs, $H$-minor free graphs with some fixed $H$. 
The second algorithm enumerates all the dominating sets in constant time per solution
for input graphs with girth at least nine. 

%% file: intro.tex
\section{Introduction}
\label{sec:intro}

One of the fundamental tasks in computer science is 
to enumerate all subgraphs satisfying a given constraint 
such as cliques~\cite{Makino:Uno:SWAT:2004}, spanning trees~\cite{Shioura:Tamura:SICOMP:1997}, 
cycles~\cite{Birmele:Ferreira:SODA:2013}, and so on. 
One of the approaches to solve enumeration problems is to design exact exponential algorithms, i.e., \name{input-sensitive} algorithms.
Another mainstream of solving enumeration problems is to design \name{output-sensitive} algorithms, i.e., the computation time depends on the sizes of both of an input and an output. 
An algorithm $\mathcal{A}$ is \name{output-polynomial} if the total computation time is polynomial of the sizes of input and output. 
$\mathcal{A}$ is an \name{incremental polynomial time algorithm} if 
the algorithm needs $\order{poly(n, i)}$ time when the algorithm outputs the $i^\text{th}$ solution after outputting the $(i-1)^\text{th}$ solution, 
where $poly(\cdot)$ is a polynomial function. 
$\mathcal{A}$ runs in \name{polynomial amortized time} if the the total computation time is $\order{poly(n)N}$, 
where $n$ and $N$ are respectively the sizes of an input and an output. 
In addition, 
$\mathcal{A}$ runs in \name{polynomial delay} if the maximum interval between two consecutive solutions is $\order{poly(n)}$ time and 
the preprocessing and postprocessing time is $\order{poly(n)}$. 
From the point of view of tractability, 
efficient algorithms for enumeration problems have been 
widely studied
~\cite{
Avis:Fukuda:DAM:1996,
Garey:Johnson:BOOK:1990,
Kante:Limouzy:WADS:2015,
Alessio:Roberto:ICALP:2016,
Wasa:Arimura:Uno:ISAAC:2014,
Eppstein:Loffler:EXP:2013,
Shioura:Tamura:SICOMP:1997,
Makino:Uno:SWAT:2004,
Tsukiyama:Shirakawa:JACM:1980,
Uno:WADS:2015,
Ferreira:Grossi:ESA:2014,
Birmele:Ferreira:SODA:2013,
Cohen:Kimefeld:Sagiv:JCSS:2008
}. 
On the other hands, 
Lawler~\textit{et al.} show that 
some enumeration problems have no output-polynomial time algorithm unless $P=NP$~\cite{Lawler:Lenstra:Rinnooy:SICOMP:1980}. 
In addition, recently, 
Creignou~\textit{et al.} show a tool for showing the hardness of enumeration problems~\cite{Creignou:Kroll:Pichler:Stritek:Vollmer:LATA:2017}. 

A dominating set is one of a fundamental substructure of graphs and
finding the minimum dominating set problem is a classical NP-hard problem~\cite{Garey:Johnson:BOOK:1990}. 
A vertex set $D$ of a graph $G$ is a dominating set of $G$ 
if every vertex in $G$ is in $D$ or has at least one neighbors in $D$. 
The enumeration of \emph{minimal} dominating sets of a graph is closely related to 
the enumeration of \emph{minimal hypergraph transversals} of a hypergraph~\cite{Eiter:Gottlob:Makino:SICOMP:2003}. 
Kant{\'e} \textit{et al}.~\cite{Kante:Limouzy:SIAM:2014} show
that the minimal dominating set enumeration problem and the minimal hypergraph transversal enumeration problem 
are equivalent, that is, the one side can be solved in output-polynomial time if 
the other side can be also solved in output-polynomial time.
Several algorithms that run in polynomial delay 
have been developed when we restrict input graphs, such as 
permutation graphs~\cite{Kante:Limouzy:SIAM:2014}, 
chordal graphs~\cite{Kante:Limouzy:WG:2015}, 
line graphs~\cite{Kante:Limouzy:WADS:2015}, 
graphs with bounded degeneracy~\cite{Kante:Limouzy:FCT:2011},
graphs with bounded tree-width~\cite{Courcelle:DAM:2009}, 
graphs with bounded clique-width~\cite{Courcelle:DAM:2009}, and 
graphs with bounded (local) LMIM-width~\cite{Golovach:Heggernes:Algo:2018}.
Incremental polynomial-time algorithms have also been developed, such as
chordal bipartite graphs~\cite{Golovach:Heggernes:Elsevier:2016}, 
graphs with bounded conformality~\cite{Boros:Elbassioni:LATIN:2004}, and
graphs with girth at least seven~\cite{Golovach:Heggernes:Algorithmica:2015}. 
Kant\'{e} \textit{et al.}~\cite{Kante:Limouzy:ISAAC:2012} show 
that the conformality of the closed neighbourhood hypergraphs of line graphs, path graphs, and ($C_4$, $C_5$, claw)-free graphs 
is constant. 
However, it is still open whether 
there exists an output-polynomial time algorithm for enumerating minimal dominating sets from general graphs.

Since the number of solutions exponentially increases compared to the minimal version, 
even if we can develop an enumeration algorithm that runs in constant time per solution, 
the algorithm becomes theoretically much slower than some enumeration algorithm for minimal dominating sets. 
However, 
when we consider the real-world problem,
we sometimes use another criteria for enumerating solutions that form dominating sets in a graph. 
That is, 
enumeration algorithms for minimal dominating sets may not fit in with other variations of minimal domination problems.  
E.g., 
a tropical dominating set~\cite{Auraiac:Bujtas:WALCOM:2016} and a rainbow dominating set~\cite{Brevsar:Henning:TJM:2008}
are such a dominating set. 
Thus, when we enumerate solutions of such domination problems, 
our algorithm becomes a base-line algorithm for these problems. 
Thus, our main goal is to develop an efficient enumeration algorithm for dominating sets.

\noindent
\textbf{Main results:}
In this paper, we consider the relaxed problems,
i.e., enumeration of all dominating sets that include non-minimal ones in a graph.
We present two algorithms, $\EnumDSD$ and $\EnumDSG$. 
$\EnumDSD$ enumerates all dominating sets in $\order{k}$ time per solution, 
where $k$ is the degeneracy of a graph (Theorem~\ref{theo:dsd:time}). 
Moreover, $\EnumDSG$ enumerates all dominating sets in constant time per solution 
for a graph with girth at least nine (Theorem~\ref{theo:girth}),
where the girth is the length of minimum cycle in the graph.

By straightforwardly using an enumeration framework such as the reverse search technique~\cite{Avis:Fukuda:DAM:1996}, 
we can obtain an enumeration algorithm for the problem that runs in $\order{n}$ or $\order{\Delta}$ time per solution,
where $n$ and $\Delta$ are respectively the number of vertices and the maximum degree of an input graph. 
Although dominating sets are fundamental in computer science, 
no enumeration algorithm for dominating sets that runs in strictly faster than such a trivial algorithm 
has been developed so far.  
Thus, to develop efficient algorithms, 
we focus on the \name{sparsity} of graphs as being a good structural property and, 
in particular, on the \name{degeneracy} and \name{girth}, which are the measures of sparseness.
As our contributions, 
we develop two optimal algorithms for enumeration of dominating sets in a sparse graph. 
We first focus on the degeneracy of an input graph.  
A graph is $k$-degenerate~\cite{Lick:CJM:1970} 
if any subgraph of the graph has a vertex whose degree is at most $k$. 
The degeneracy of a graph is the minimum value of $k$ such that the graph is $k$-degenerate. 
Note that $k \le \Delta$ always holds. 
It is known that some graph classes have constant degeneracy, 
such as forests, grid graphs, outerplanar graphs, planer graphs, bounded tree width graphs, 
and $H$-minor free graphs for some fixed $H$~\cite{Thomason:JCT:2001,Chandran:Subramanian:JCombin:2005}. 
A $k$-degenerate graph has a \emph{good} vertex ordering, 
called a \name{degeneracy ordering}~\cite{Matula:Beck:JACM:1983},
as shown in Section~\ref{subsec:deg}. 
So far, 
this ordering has been used to develop 
efficient enumeration algorithms~\cite{Alessio:Roberto:ICALP:2016,Wasa:Arimura:Uno:ISAAC:2014,Eppstein:Loffler:EXP:2013}. 
By using this ordering and the reverse search technique~\cite{Avis:Fukuda:DAM:1996}, 
we show that our proposed algorithm \EnumDSD can solve the relaxed problem in $\order{k}$ time per solution.  
This implies that \EnumDSD can optimally enumerate all the dominating sets in an input graph with constant degeneracy. 

We next focus on the girth of a graph. 
Enumeration of minimal dominating sets can be solved efficiently if an input graph has no short cycles 
since its connected subgraphs with small diameter form a tree. 
Indeed, this local tree structure has been used in minimal dominating sets enumeration~\cite{Golovach:Heggernes:Algorithmica:2015}. 
For the relaxed problem, 
by using the reverse search technique, 
we can easily show that the delay of our proposed algorithm \EnumDSG for general graphs is $\order{\Delta^3}$ time. 
However, if an input graph has the large girth, 
then each recursive call generates enough solutions, that is, we can amortize the complexity of \EnumDSG. 
Thus, by amortizing the time complexity using this local tree structure, 
we show that the problem can be solve in constant time per solution for graphs with girth at least nine. 

%% file: basic.tex
yy
\section{A Basic Algorithm Based on Reverse Search}
\label{sec:enum}

\input{prelim.tex}

In this paper, we propose two algorithms \EnumDSD and \EnumDSG for solving Problem~\ref{prob:main}. 
These algorithms use the degeneracy ordering and the local tree structure, respectively. 
Before we enter into details of them, 
we first show the basic idea for them, called \name{reverse search method} that is proposed by Avis and Fukuda~\cite{Avis:Fukuda:DAM:1996}
and is one of the framework for constructing enumeration algorithms. 


An algorithm based on reverse search method enumerates solutions
by traversing on an implicit tree structure on the set of solution, called a \name{family tree}.  
For building the family tree, 
we first define the parent-child relationship between solutions as follows: 
Let  $G = (V, E)$ be an input graph with $V = \set{v_1, \dots, v_n}$ and  
$X$ and $Y$ be dominating sets on $G$. 
We arbitrarily number the vertices in $G$ from $1$ to $n$ and 
call the number of a vertex the \name{index} of the vertex. 
If no confusion occurs, 
we identify a vertex with its index. 
We assume that there is a total ordering $<$ on $V$ according to the indices. 
$\pv{X}$, called the \name{parent vertex},
is the vertex in $V \setminus X$ with the minimum index.
For any dominating set $X$ such that $X \neq V$, 
$Y$ is the \name{parent} of $X$ if $Y = X \cup \set{\pv{X}}$. 
We denote by $\Parent{X}$ the parent of $X$. 
Note that since any superset of a dominating set also dominates $G$, 
thus,  
$\Parent{X}$ is also a dominating set of $G$. 
We call $X$ is a \name{child} of $Y$ if $\Parent{X}= Y$. 
We denote by $\Fam{G}$ a digraph on the set of solutions $\Sol{G}$. 
Here, the vertex set of $\Fam{G}$ is $\Sol{G}$ and
the edge set $\PC{G}$ of $\Fam{G}$ is defined according to the parent-child relationship. 
We call $\Fam{G}$ the \name{family tree} for $G$ and 
call $V$ the \name{root} of $\Fam{G}$. 
Next, we show that $\Fam{G}$ forms a tree rooted at $V$. 

Our basic algorithm \EnumDS is shown in Algorithm~\ref{algo:naive}. 
We say $\cand{X}$ the \name{candidate set} of $X$ and 
define  $\cand{X} = \inset{v \in V}{N[X \setminus \set{v}] = V \land \Parent{X\setminus\set{v}} = X}$. 
Intuitively, the candidate set of $X$ is the set of vertices such that any vertex $v$ in the set, 
removing $v$ from $X$ generates another dominating set. 
We show a recursive procedure \AllChildren{$X, \cand{X}, G$}
actually  generates all children of $X$ on $\Fam{G}$. 
We denote by $ch(X)$ the set of children of $X$, 
    and by $gch(X)$ the set of grandchildren of $X$. 
    
From Lemmas~\ref{lem:connected}, \ref{lem:tree}, and ~\ref{lem:enum}, 
we can obtain the correctness of \EnumDS. 
    
\begin{algorithm}[t]
  \caption{\texttt{EDS} enumerates all dominating sets in amortized polynomial time. }
  \label{algo:naive}
  \Procedure(\tcp*[f]{$G$: an input graph}){\EnumDS{$G = (V, E)$}}{ 
    \AllChildren{$V, V, G$}\;
  }
  \Procedure(\tcp*[f]{$X$: the current solution}){\AllChildren{$X, \cand{X}, G = (V, E)$}}{
    Output $X$\label{step:start}\;
    \For(){$v \in \cand{X}$}{
        $Y \gets X \setminus \set{v}$;  $\cand{Y} \gets \inset{u \in \cand{X}}{N[Y\setminus\set{u}] = V \land \Parent{Y \setminus\set{u}} = Y}$\;
      \AllChildren{$Y, \cand{Y}, G$}\;\label{step:recursive:zero}        
    }
    \Return \;
 }
\end{algorithm}

\begin{lemma}
  \label{lem:connected}
  For any dominating set $X$, 
  by recursively applying the parent function $\Parent{\cdot}$ to $X$ at most $n$ times, 
  we obtain $V$. 
\end{lemma}

\begin{proof}
  For any dominating set $X$, 
  since $\pv{v}$ always exists, 
  there always exists the parent vertex for $X$. 
  In addition,  $\size{\Parent{X} \setminus X} = 1$. 
  Hence, the statement holds. 
\end{proof}

\begin{lemma}
    \label{lem:tree}
    $\Fam{G}$ forms a tree. 
\end{lemma}

\begin{proof}
    Let $X$ be any solution in $\Sol{G}\setminus\set{V}$.
    Since $X$ has exactly one parent and $V$ has no parent,  
    $\Fam{G}$ has $\size{V(\Fam{G})} - 1$ edges. 
    In addition, since there is a path between $X$ and $V$ by Lemma~\ref{lem:connected}, 
    $\Fam{G}$ is connected. 
    Hence,  the statement holds. 
\end{proof}

\begin{lemma}
\label{lem:enum}
    Let $X$ and $Y$ be distinct dominating sets in a graph $G$. 
    $Y \in ch(X)$ if and only if there is a vertex $v \in \cand{X}$ 
    such that $X = Y \cup \set{v}$.
\end{lemma}

\begin{proof}
    The if part is immediately shown from the definition of a candidate set. 
    We show the only if part by contradiction. 
    Let $Z$ be a dominating set in $ch(X)$ such that $Z = X\setminus\set{v'}$, where $v' \in Z$. 
    We assume that $v' \notin \cand{X}$. 
    From $v' \notin \cand{X}$, $N[\Parent{Z}] \neq V$ or $\Parent{Z} \neq X$. 
    Since $Z$ is a child of $X$, 
    $\Parent{Z} = X$, and thus, $N[\Parent{Z}] = V$. 
    This contradicts $v' \notin \cand{X}$. 
    Hence,  the statement holds.
\end{proof}

\begin{theorem}
  \label{theo:enum}
  By traversing $\Fam{G}$, 
  \EnumDS solves Problem~\ref{prob:main}.
\end{theorem}

%% file: prelim.tex

Let $G = (V(G), E(G))$ be a simple undirected graph, that is, $G$ has no self loops and multiple edges,  
with vertex set $V(G)$ and edge set $E(G)$ is a set of pairs of vertices. 
If no confusion arises, we will write $V = V(G)$ and $E = E(G)$. 
Let $u$ and $v$ be vertices in $G$. 
An edge $e$ with $u$ and $v$ is denoted by $e = \set{u, v}$. 
$u$ and $v$ are \name{adjacent} if  $\set{u, v} \in E$.
We denote by $N_G(u)$ the set of vertices that are adjacent to $u$ on $G$  
and by $N_G[u] = N_G(u) \cup \set{u}$.
We say $v$ is a \name{neighbor} of $u$ if $v \in N_G(u)$. 
The \name{set of neighbors} of $U$ is defined as 
$N(U) = \bigcup_{u \in U}N_G(u) \setminus U$.
Similarly, let $N[U]$ be $\bigcup_{u \in U}N_G(u) \cup U$.
Let $d_G(v) = \size{N_G(v)}$ be the \name{degree} of $u$ in $G$. 
We call the vertex $v$ \name{pendant} if $d_G(v) = 1$.
$\Delta(G) = \max_{v \in V} d(v)$ denotes the maximum degree of $G$. 
A set $X$ of vertices is a \name{dominating set} 
if $X$ satisfies $N[X] = V$.

For any vertex subset $V' \subseteq V$, 
we call $G[V'] = (V', E[V'])$ an \name{induced subgraph} of $G$, 
where $E[V'] = \inset{\set{u, v} \in E(G)}{u, v \in V'}$.
Since $G[V']$ is uniquely determined by $V'$, 
we identify $G[V']$ with $V'$. 
We denote by $G \setminus \set{e} = (V, E \setminus \set{e})$ and 
$G\setminus \set{v} = G[V\setminus\set{v}]$.
For simplicity, we will use $v \in G$ and $e \in G$ to refer to
$v \in V(G)$ and $e \in E(G)$, respectively.

We now define the dominating set enumeration problem as follows:  
\begin{problem}
    \label{prob:main}
    Given a graph $G$, then 
    output all dominating sets in $G$ without duplication. 
\end{problem}

%% file: degenerate.tex
\section{Efficient Enumeration for Bounded Degenerate Graphs}
\label{subsec:deg}

\input{enumdsd_short}

\begin{figure}[t]
    \centering
    \includegraphics[width=0.6\textwidth]{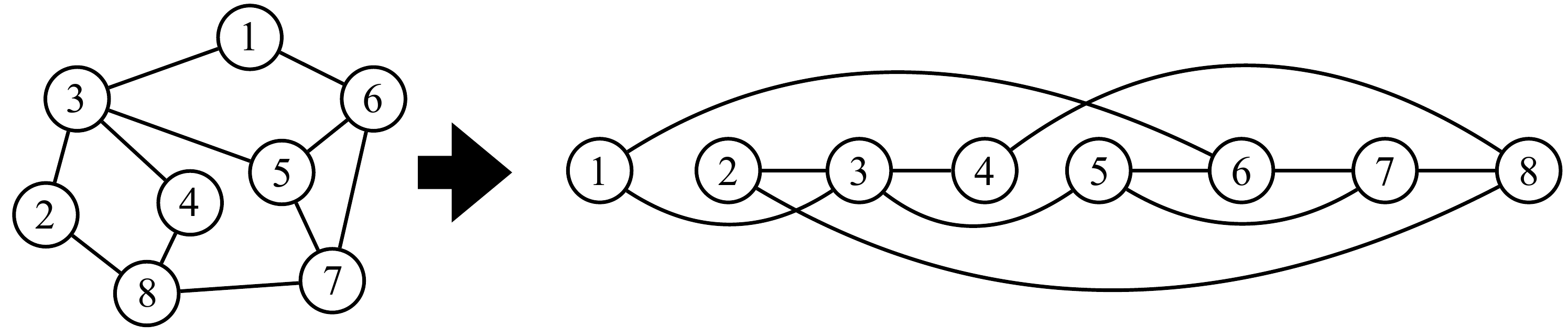}
    \caption{An example of a degeneracy ordering for a $2$-degenerate graph $G$. 
    In this ordering, each vertex $v$ is adjacent to 
    vertices at most two whose indices are larger than $v$.  }
    \label{fig:dege}
\end{figure}

The bottle-neck of \EnumDS is the maintenance of candidate sets. 
Let $X$ be a dominating set and $Y$ be a child of $X$. 
We can easily see that the time complexity of \EnumDS is $\order{\Delta^2}$ time per solution 
    since a removed vertex $u \in \cand{X}\setminus\cand{Y}$ has the distance at most two from $v$. 
In this section, we improve \EnumDS by focusing on the degeneracy of an input graph $G$. 
$G$ is a \name{$k$-degenerate graph}~\cite{Lick:CJM:1970} if 
for any induced subgraph $H$ of $G$, the minimum degree in $H$ is less than or  equal to $k$. 
The \name{degeneracy} of $G$ is the smallest $k$ such that $G$ is $k$-degenerate. 
A $k$-degenerate graph has a \emph{good} vertex ordering. 
The definition of orderings of vertices in $G$,  called a \name{degeneracy ordering} of $G$, 
is as follows: 
for any vertex $v$ in $G$, 
the number of vertices that are larger than $v$ and adjacent to $v$ is at most $k$. 
We show an example of a degeneracy ordering of a graph in Fig.~\ref{fig:dege}. 
Matula and Beck show that 
the degeneracy and a degeneracy ordering of $G$ can be obtained in $\order{n + m}$ time~\cite{Matula:Beck:JACM:1983}. 
Our proposed algorithm \EnumDSD, shown in Algorithm~\ref{algo:dsd}, achieves amortized $\order{k}$ time enumeration by using this good ordering. 
In what follows, 
we fix some degeneracy ordering of $G$ and number the indices of vertices from $1$ to $n$ according to the degeneracy ordering. 
We assume that for each vertex $v$ and each dominating set $X$, 
$N[v]$ and $\cand{X}$ are stored in a doubly linked list and sorted by the ordering. 
Note that the larger neighbors of $v$ can be listed in $\order{k}$ time. 
Let us denote by $V^{<v} = \set{1,2, \dots, v-1}$ and $V^{v<} = \set{v+1, \dots n}$. 
Moreover, $A^{<v} = A \cap V^{v<}$
and $A^{v<} =  A \cap V^{<v}$  for a subset $A$ of $V$. 
We first show the relation between $\cand{X}$ and $\cand{Y}$. 

\begin{lemma}
\label{lem:subset}
    Let $X$ be a dominating set of $G$ and $Y$ be a child of $X$. 
    Then, $\cand{Y} \subset \cand{X}$. 
\end{lemma}

\begin{proof}
Let $Z$ be a child of $Y$. 
Hence, $\pv{Z} \in X$ and $\pv{Z} \in \cand{Y}$.  
From the definition of $\pv{Z}$, 
$\pv{Z} = \min{V \setminus Z}$. 
Moreover,  since $V \setminus X \subset V \setminus Z$, 
$\pv{Z} \le \min{V\setminus X}$. 
Therefore, 
$\pv{Z} \in \cand{X}$. 
\end{proof}

From the Lemma~\ref{lem:subset}, 
for any $v \in \cand{X}$, 
what we need to obtain the candidate set of $Y$ is to compute $\Del{X}{\pv{Y}} = \cand{X}\setminus\cand{Y}$,
where $Y = X \setminus \set{v}$. 
In addition, we can easily sort $\cand{Y}$ by the degeneracy ordering if $\cand{X}$ is sorted. 
In what follows, we denote by  
    $\Del[1]{X}{v} = \inset{u \in \cand{X}^{<v}}{N[u] \cap X = \set{u, v}}$, 
    $\Del[2]{X}{v} = \inset{u \in \cand{X}^{<v}}{\exists w \in V\setminus (X \setminus \set{v}) (N[w] \cap X = \set{u, v})}$, and 
    $\Del[3]{X}{v} = \cand{X}^{v \le}$. 
Next, 
    we show the time complexity for obtaining $\Del{X}{\pv{Y}}$. 

\begin{lemma}
\label{lem:del}
    For each $v \in \cand{X}$,  
    $\Del{X}{v} = \Del[1]{X}{v} \cup \Del[2]{X}{v} \cup \Del[3]{X}{v}$ holds.  
\end{lemma}

\begin{proof}
    $\Del{X}{v} \supseteq \Del[1]{X}{v} \cup \Del[2]{X}{v} \cup \Del[3]{X}{v}$ is trivial 
    since $X \setminus\set{u, v}$ is not dominating set for each $u \in \Del[1]{X}{v} \cup \Del[2]{X}{v}$ 
    and the parent of $X \setminus \set{u, v}$ is not $X\setminus \set{v}$ for each $u \in \Del[3]{X}{v}$. 
    We next prove $\Del{X}{v} \subseteq \Del[1]{X}{v} \cup \Del[2]{X}{v} \cup \Del[3]{X}{v}$. 
    Let $u$ be a vertex in $\Del{X}{v}$. 
    Suppose that $X\setminus\set{u, v}$ is a dominating set. 
    Since $\Parent{X\setminus\set{u, v}} \neq X \setminus \set{v}$, $v < u$. 
    Thus, $u \in \Del[3]{X}{v}$. 
    Suppose that $X\setminus\set{u, v}$ is not a dominating set, 
    that is,  $N[X \setminus\set{u, v}] \neq V$. 
    This implies that there exists a vertex $w$ in $V$ 
        such that $w$ is not dominated by any vertex in $X \setminus\set{u, v}$.  
    Note that $w$ may be equal to $u$ or $v$. 
    Hence, $N[w] \cap X = \set{u, v}$ and the statement holds. 
\end{proof}

\begin{figure}
    \centering
    \includegraphics[width=0.7\textwidth]{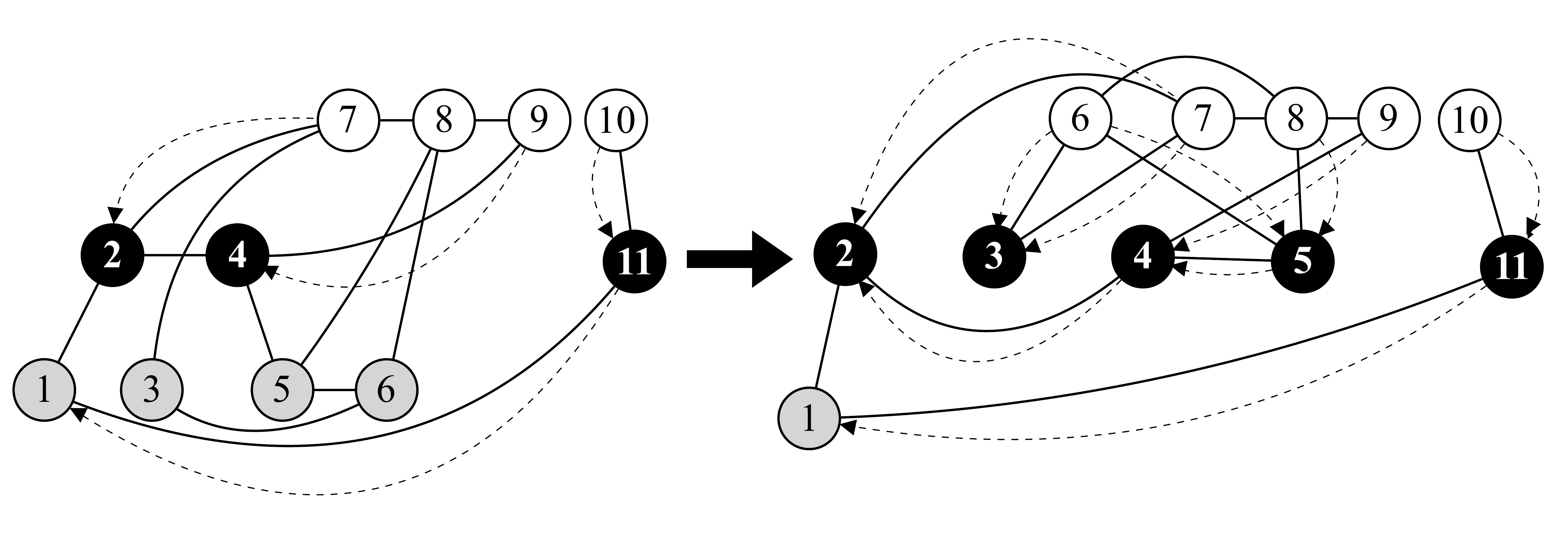}
    \caption{
    Let $X$ be a dominating set $\set{1,2,3,4,5,6,11}$. 
    An example of the maintenance of $\cand{X}$ and $\sig D(X)$. 
    Each dashed directed edge is stored in $\sig D(X)$, and
    each solid edge is an edge in $G$. 
    A directed edge $(u, v)$ implies $v \in D_u(X)$. 
    The index of each vertex is according to a degeneracy ordering.
    White, black, and gray vertices belong to $V\setminus X$, $X\setminus\cand{X}$, and $\cand{X}$, respectively. 
    When \texttt{EDS-D} removes vertex $6$, $\cand{X\setminus\set{6}} = \set{1}$.
        }
    \label{fig:edsd}
\end{figure}

We show an example of dominated list and a maintenance of $\cand{X}$ in Fig.~\ref{fig:edsd}. 
To compute a candidate set efficiently,
for each vertex $u$ in $V$,  
we maintain the vertex lists $D_u(X)$ for $X$. 
We call $D_u(X)$ the \name{dominated list} of $u$ for $X$. 
The definition of $D_u(X)$ is as follows: 
If $u \in V \setminus X$, then $D_u(X) =  N(u) \cap (X \setminus \cand{X})$. 
If $u \in X$, then $D_u(X) = N(u)^{<u} \cap (X \setminus \cand{X})$.
For brevity, we write $D_u$ as $D_u(X)$ if no confusion arises. 
We denote by $\sig D(X) = \bigcup_{u \in V} \set{D_u}$. 
By using $\sig D(X)$, we can efficiently find $\Del[1]{X}{v}$ and $\Del[2]{X}{v}$. 

\begin{lemma}
\label{lem:prepro}
    For each vertex $v \in \cand{X}$, 
    we can compute $N(v) \cap \cand{X}$ and $N(v)^{v<} \cap X$ in $\order{k}$ time on average over all children of $X$. 
\end{lemma}

\begin{proof}
    Since $G$ is $k$-degenerate, 
    $G[\cand{X}]$ is also $k$-degenerate. 
    Thus, the number of edges in $G[\cand{X}]$ is at most $k\size{\cand{X}}$. 
    Remind that $\cand{X}$ is sorted by the degeneracy ordering. 
    Hence,  by scanning vertices of $\cand{X}$ from the smallest vertex to the largest one, 
    for each $v$ in $\cand{X}$, 
    we can obtain $N(v) \cap \cand{X}$ in $\order{k}$ time on average over all children of $X$. 
    Since $N(v)^{v<}$ is the larger $v$'s neighbors set, 
    the size is at most $k$. 
    Hence, the statement holds.
\end{proof}

\begin{lemma}
\label{lem:time1}
    Let $X$ be a dominating set of $G$. 
    Suppose that for each vertex $u$ in $G$, 
        we can obtain the size of $D_u$ in constant time. 
    Then, for each vertex $v \in \cand{X}$,
    we can compute $\Del[1]{X}{v}$ in $\order{k}$ time on average over all children of $X$. 
\end{lemma}

\begin{proof}
    Since every vertex $u$ in $\Del[1]{X}{v}$ is adjacent to $v$, 
        $\Del[1]{X}{v} \subseteq N(v) \cap \cand{X}$. 
    To compute $\Del[1]{X}{v}$, we need to check whether $N[u]\cap X = \set{u, v}$ or not.
    We first consider smaller neighbors of $u$. 
    Before computing $\Del[1]{X}{v}$ for every vertex $v$, 
        we record the size of $D_u$ of $u \in \cand{X}$ in $\order{\size{\cand{X}}}$ time. 
    $D_u =\emptyset$ if and only if there are no smaller neighbors of $u$ in $X^{<u}\setminus\cand{X}$. 
    Moreover, the number of edges in $G[\cand{X}]$ is at most $k\size{\cand{X}}$ from the definition of the degeneracy. 
    Thus, this part can be done in $\order{\sum_{v \in \cand{X}}\size{N(v)\cap\cand{X}}}$ total time and 
        in $\order{k}$ time per each vertex in $\cand{X}$. 
    We next consider larger neighbors. 
    Again, before computing $\Del[1]{X}{v}$ for every vertex $v$, 
    from Lemma~\ref{lem:prepro} and the degeneracy of $G$, 
        we can check all of the larger neighbors of $u \in \cand{X}$ in $\order{k\size{\cand{X}}}$ time. 
    Thus, as with the smaller case, 
        the checking for the larger part also can be done in $\order{k}$ time on average over all children of $X$. 
    Hence, the statement holds. 
\end{proof}

\begin{lemma}
\label{lem:time2}
    Suppose that for each vertex $w$ in $G$, 
        we can obtain the size of $D_w$ in constant time. 
    For each vertex $v \in \cand{X}$,
        we can compute $\Del[2]{X}{v}$ in $\order{k}$ time on average over all children of $X$. 
\end{lemma}

\begin{proof}
    Let $u$ be a vertex in $\Del[2]{X}{v}$. 
    Then, there exists a vertex $w$ 
        such that $N[w]\cap X = \set{u, v}$ and $w \in  N[v] \cap (V \setminus (X \setminus \set{v})$. 
    In addition, for any vertex $v'$ in $\cand{X}$, 
        $\pv{X \setminus \set{v'}} = v'$. 
    Thus, $v \le w$ and $u < w$ hold. 
    Before computing $\Del[2]{X}{v}$ for every vertex $v$, 
        by scanning all larger neighbors $w'$ of vertices of $\cand{X}$, 
            we can list such vertices $w'$ such that
                $w' > \max\set{\cand{X}}$, $\size{N[w'] \cap \cand{X}} = 2$, and $w' \in V \setminus (X \setminus \set{v})$ in $\order{k\size{\cand{X}}}$ time 
                    since $G$ is $k$-degenerate. 
    If $D_{w'} \neq \emptyset$, that is, $w'$ has a neighbor in $X\setminus\cand{X}$, 
        then $\size{N[w] \cap X} > 2$. 
    Thus, since we can check the size of $D_{w'}$ in constant time, 
        we can compute $\Del[2]{X}{v}$ in $\order{k}$ time on average over all children of $X$. 
\end{proof}

In Lemma~\ref{lem:time1} and Lemma~\ref{lem:time2}, 
we assume that the dominated lists were computed when we compute $\Del{X}{v}$ for each vertex $v$ in $\cand{X}$. 
We next consider how we maintain $\sig D$. 
Next lemmas show the transformation from $D_u(X)$ to $D_u(Y)$ for each vertex $u$ in $G$.

\newcommand{\XbC}{X_{\bar{C}}}
\begin{lemma}
\label{lem:adj1}
    Let $X$ be a dominating set, $v$ be a vertex in $\cand{X}$, and $Y = X \setminus \set{v}$.
    For each vertex $u \in G$ such that $u \neq v$, 
    $D_u(Y) = D_u(X) \cup (N(u)^{<u} \cap (\Del[1]{X}{v} \cup\Del[2]{X}{v})) \cup (N(u)^{<u}\cap(\Del[3]{X}{v}\setminus\set{v}))$. 
\end{lemma}

\begin{proof}
    Let $\XbC= X\setminus\cand{X}$. 
    Suppose that $u \in Y$. 
    From the definition, $D_u(X) = N(u)^{<u} \cap\XbC$. 
    From the distributive property, 
    \begin{eqnarray*}
        L & = & D_u(X) \cup (N(u)^{<u} \cap (\Del[1]{X}{v} \cup\Del[2]{X}{v})) \cup (N(u)^{<u}\cap(\Del[3]{X}{v}\setminus\set{v})) \\
          & = & N(u)^{<u}\cap(\XbC\cup(\Del{X}{v}\setminus\set{v})) \\
          & = & N(u)^{<u}\cap(Y\setminus\cand{Y})
    \end{eqnarray*}
    Since $\XbC\cup(\Del{X}{v}\setminus\set{v}) = Y\setminus\cand{Y}$. 
    Suppose that $u \in V \setminus X$. 
    From the parent-child relation, $\pv{Y} < u$ holds. 
    Since $\Del{X}{v} \subseteq V^{<u}$, 
        $N(u)^{<u} \cap(\Del[1]{X}{v}\cup\Del[2]{X}{v}) = N(u)\cap(\Del[1]{X}{v}\cup\Del[2]{X}{v})$, and
        $N(u)^{<u} \cap(\Del[3]{X}{v}\setminus\set{v}) = N(u) \cap(\Del[3]{X}{v}\setminus\set{v})$.
    From the definition, $D_u(X) = N(u) \cap\XbC$, 
    \begin{eqnarray*}
        L & = & D_u(X) \cup (N(u)^{<u} \cap (\Del[1]{X}{v} \cup\Del[2]{X}{v})) \cup (N(u)^{<u}\cap(\Del[3]{X}{v}\setminus\set{v})) \\
          & = & (N(u) \cap \XbC) \cup (N(u)\cap(\Del[1]{X}{v}\cup\Del[2]{X}{v})) \cup (N(u)\cap(\Del[3]{X}{v}\setminus\set{v})) \\
          & = & N(u)\cap(\XbC \cup(\Del[1]{X}{v} \cup\Del[2]{X}{v})\cup(\Del[3]{X}{v}\setminus\set{v}))\\
          & = & N(u)\cap(\XbC\cup(\Del{X}{v}\setminus\set{v}))\\
          & = & N(u)\cap(Y\setminus\cand{Y})
    \end{eqnarray*}
    Hence, the statement holds. 
\end{proof}

\begin{lemma}
\label{lem:adj2}
    Let $X$ be a dominating set, $v$ be a vertex in $\cand{X}$, and  $Y = X \setminus \set{v}$.
    $D_v(Y) = D_v(X) \cup (N(v)^{<v} \cap (\Del[1]{X}{v}\cup\Del[2]{X}{v})) \cup (N(v)^{v<} \cap X)$.
\end{lemma}

\begin{proof}
    Since $\Del[1]{X}{v} \cup \Del[2]{X}{v} \subseteq V^{<v}$ and $\Del[3]{X}{v}\cap V^{<v} = \emptyset$,
    $N(v)^{<v} \cap (\Del[1]{X}{v}\cup\Del[2]{X}{v}) = N(v)^{<v}\cap\Del{X}{v}$. 
    By the same discussion as Lemma~\ref{lem:adj1},
    $L = D_v(X) \cup (N(v)^{<v} \cap \Del{X}{v})= N(v)^{<v} \cap(Y\setminus\cand{Y})$.
    Since $Y=X\setminus\set{v}$, $N(v)^{v<}\cap Y = N(v)^{v<}\cap X$. 
    Moreover, since $X^{<v} = Y^{<v}$ and $\cand{Y}^{v<} = \emptyset$, 
    $N(v)^{v<} \cap(Y\setminus\cand{Y}) = N(v)^{v<} \cap X$. 
    Since $L = (N(v)^{v<} \cup N(v)^{<v})\cap(Y\setminus\cand{Y}) = D_v(Y)$, 
    the statement holds.
\end{proof}


We next consider the time complexity for obtaining the dominated lists for children of $X$. 
From Lemma~\ref{lem:adj1} and Lemma~\ref{lem:adj2}, 
    a na\"{i}ve method for the computation needs $\order{k\size{\Del{X}{v}} + k}$ time for each vertex $v$ of $X$ 
    since we can list all larger neighbors of any vertex in $\order{k}$ time. 
However, 
if we already know $\cand{W}$ and $\sig D(W)$ for a child $W$ of $X$, 
    then we can easily obtain $\sig D(Y)$, where $Y$ is the child of $X$ immediately after $W$. 
The next lemma plays a key role in \EnumDSD. 
Here, for any two sets $A, B$, 
    we denote by $A \oplus B = (A \setminus B) \cup (B \setminus A)$. 

\begin{lemma}
\label{lem:compadj}
    Let $X$ be a dominating set, 
    $v, u$ be vertices in $\cand{X}$ such that $u$ has the maximum index in $\cand{X}^{<v}$, 
    $Y = X\setminus\set{u}$, and $W = X \setminus \set{v}$. 
    Suppose that we already know $\cand{Y}\oplus\cand{W}$, $\sig D(W)$, $\Del{X}{v}$, and $\Del{X}{u}$. 
    Then, we can compute $\sig D(Y)$ in $\order{k\size{\cand{Y}\oplus\cand{W}} + k}$ time.  
\end{lemma}
\begin{proof}
    Suppose that $z$ is a vertex in $G$ such that $z \neq v$ and $z \neq u$. 
    From the definition,
    $D_z(W) \setminus D_z(Y) = (\Del{X}{v} \setminus \Del{X}{u}) \cap N(z)^{<z}$ and 
    $D_z(Y) \setminus D_z(W) = (\Del{X}{u} \setminus \Del{X}{v}) \cap N(z)^{<z}$. 
    Hence, 
    we first compute $\Del{X}{v} \oplus \Del{X}{u}$. 
    Now, $(\cand{X} \setminus \cand{W}) \oplus (\cand{X} \setminus \cand{Y}) 
        = \cand{W} \oplus \cand{Y}$. 
    Next, 
    for each vertex $c$ in $\cand{W} \oplus \cand{Y}$, 
    we check whether we add to or remove $c$ from $D_z(Y)$ or not. 
    Note that added or removed vertices from $D_z(Y)$ is a smaller neighbor of $z$. 
    From the definition, 
    if $c \notin D_z(Y)$ or $c \in D_z(X)$, then we add $c$ to $D_z(Y)$.
    Otherwise, we remove $c$ from $D_z(Y)$. 
    Thus, since each vertex in $\cand{W} \oplus \cand{Y}$ has at most $k$ larger neighbors, 
    for all vertices other than $u$ and $v$, 
    we can compute the all dominated lists in $\order{k\size{\cand{W} \oplus \cand{Y}}}$ time. 
    Next we consider the update for $D_u(Y)$ and $D_v(Y)$. 
    Note that from the definition, 
    $D_v(W)$ and $D_u(Y)$ contain larger neighbors of $v$ and $u$, respectively. 
    However, the number of such neighbors is $\order{k}$. 
    Finally, 
    since $v$ belongs to $Y$, 
    $v \in D_{u'}(Z)$ if $u' \in N(v)^{v<}$ for any vertex $u'$. 
    Thus, 
    as with the above discussion, 
    we can compute  $D_u(Y)$ and $D_v(Y)$ in $\order{k\size{\cand{W} \oplus \cand{Y}} + k}$ time. 
\end{proof}


\begin{lemma}
\label{lem:update}
    Let $X$ be a dominating set. 
    Then, \AllChildren{$X, \cand{X}, \sig D(X)$} of \EnumDSD  other than recursive calls can be done 
    in $\order{k\size{ch(X)} + k\size{gch(X)}}$ time. 
\end{lemma}

\begin{proof}
    We first consider the time complexity of \CandD. 
    From Lemma~\ref{lem:time1} and Lemma~\ref{lem:time2}, 
        \CandD correctly computes $\Del[1]{X}{v}$ and $\Del[2]{X}{v}$ in 
        from line~\ref{algo:dsd:candd:del1:loop:smaller:start}
        to line~\ref{algo:dsd:candd:del1:loop:smaller:end} and 
        from line~\ref{algo:dsd:candd:del1:loop:larger:start}
        to line~\ref{algo:dsd:candd:del1:loop:larger:end}, respectively. 
    For each loop from line~\ref{algo:dsd:loop}, 
        the algorithm picks the largest vertex in $C$. 
    This can be done in $\order{1}$ since $C$ is sorted. 
    The algorithm needs to remove vertices in $\Del[3]{X}{v}$. 
    This can be done in line~\ref{algo:dsd:get:C} and 
        in $\order{1}$ time since $v$ is the largest vertex. 
    Thus, for each vertex $v$ in $\cand{X}$, 
        $\cand{X\setminus\set{v}}$ can be obtained in $\order{k}$ time on average. 
    Hence, 
    for all vertices in $\cand{X}$, 
        the candidate sets can be computed in $\order{k\size{ch(X)}}$ time. 
    Next, we consider the time complexity of \DomList. 
    Before computing \DomList, 
    \EnumDSD already computed $\cand{Y}\oplus\cand{W}$, $\sig D(W)$, $\Del{X}{v}$, and $\Del{X}{v'}$. 
    Note that we can compute $\cand{Y} \oplus \cand{W}$ when we compute $\cand{Y}$ and $\cand{W}$. 
    Here, 
        $W$ is the previous dominating set, 
        $C'$ stores $\cand{W}$, and 
        $\sig D'$ stores $\sig D(W)$.  
    Thus, by using Lemma~\ref{lem:compadj}, 
    we can compute $\sig D(Y)$ in $\order{k\size{\cand{Y}\oplus\cand{W}} + k}$ time.  
    In addition, 
    for all vertices in $\cand{X}$, 
        the dominated lists can be computed in $\order{k\size{\cand{X}} + k\size{gch(X)}}$ time
        since $Y$ has at least $\size{\cand{W}\setminus\cand{Y}} - 1$ children and 
        $\size{gch(X)}$ is at least the sum of $\size{\cand{W}\setminus\cand{Y}} - 1$ over all $Y \in \inset{X \setminus \set{v}}{v \in \cand{X}}$ and the previous solution $W$ of $Y$. 
    When \EnumDSD copies data such as $\sig D$, 
        \EnumDSD only copies the pointer of these data. 
    By recording operations of each line, 
        \EnumDSD restores these data when backtracking happens. 
    These restoring can be done in the same time of the above update computation. 
\end{proof}

\begin{theorem}
  \label{theo:dsd:time}
  \EnumDSD enumerates all dominating sets in $\order{k}$ time per solution 
  in a $k$-degenerate graph by using $\order{n + m}$ space. 
\end{theorem}

\begin{proof}
    The parent-child relation of \EnumDSD and \EnumDS are same.
    From Lemma~\ref{lem:subset} and Lemma~\ref{lem:del}, \EnumDSD correctly computes all children. 
    Hence, the correctness of \EnumDSD is shown by the same manner of Theorem~\ref{theo:enum}. 
    We next consider the space complexity of \EnumDSD. 
    For any vertex $v$ in $G$, 
    if $v$ is removed from a data structure used in \EnumDSD on a recursive procedure, 
    $v$ will never be added to the data structure on descendant recursive procedures. 
    In addition, 
    for each recursive procedure, the number of data structures that are used in the procedure is constant.
    Hence, the space complexity of \EnumDSD is $\order{n + m}$. 
    We finally consider the time complexity. 
    Each recursive procedure needs $\order{k\size{ch(X)} + k\size{gch(X)}}$ time from Lemma~\ref{lem:update}. 
    Thus, 
    the time complexity of \EnumDSD is $\order{k\sum_{X \in \sig S}(\size{ch(X)} + \size{gch(X)})}$, where $\sig S$ is the set of solutions. 
    Now, $\order{\sum_{X \in \sig S}(\size{ch(X)} + \size{gch(X)})} = \order{\size{\sig S}}$.
    Hence, the statement holds.
\end{proof}

%% file: enumdsd_short.tex
\begin{algorithm}[!t]
  \caption{\texttt{EDS-D} enumerates all dominating sets in $\order{k}$ time per solution. }
  \label{algo:dsd}
    \Procedure(\tcp*[f]{$G$: an input graph}){\EnumDSD{$G = (V, E)$}}{
        \lFor{$v \in V$}{$D_v \gets \emptyset$}
        \AllChildren{$V, V, \sig D(V) := \set{D_1, \dots, D_{\size{V}}}$}\;
    }
    \Procedure{\AllChildren{$X, C, \sig D$}}{
        Output $X$\;\label{algo:dsd:startrec}
        $C'  \gets \emptyset$; 
        $\sig D' \gets \sig D$\tcp*{$\sig D' := \set{D'_1, \dots, D'_{|V|}}$}\label{algo:dsd:init}
        \For(\tcp*[f]{$v$ has the largest index in $C$}){$v \in C$}{  \label{algo:dsd:loop}
            $Y \gets X \setminus \set{v}$\; \label{algo:dsd:get:Y}
            $C \gets C \setminus \set{v}$\tcp*{Remove vertices in $\Del[3]{X}{v}$.} \label{algo:dsd:get:C}
            $\cand{Y} \gets$ \CandD{$X, v, C$}\tcp*{Vertices larger than $v$ are not in $C$.}      \label{algo:dsd:cand-d}
            $\sig D(Y) \gets$ \DomList{$v, Y, X, \cand{Y}, C'\oplus\cand{Y}, \sig D'$}\;\label{algo:dsd:domlist}
            \AllChildren{$Y, \cand{Y}, \sig D(Y)$}\; \label{algo:dsd:rec}
            $C' \gets \cand{Y}$;  $\sig D' \gets \sig D(Y)$\; \label{algo:dsd:record:previous}
            \lFor{$u \in N(v)^{v<}$}{ \label{algo:dsd:make:start}
                $D'_u \gets D'_u \cup \set{v}$\label{algo:dsd:make}
            } \label{algo:dsd:make:end}
        }
    }
    \Procedure(){\CandD{$X, v, C$}}{
        $Y \gets X \setminus \set{v}$;  $Del_1 \gets \emptyset$; $ Del_2 \gets \emptyset$\;
        \For{$u \in (N(v) \cap C) \cup N(v)^{v<}$}{ \label{algo:dsd:candd:del1:loop:start}
            \eIf{$u < v$}{ \label{algo:dsd:candd:del1:loop:smaller:start}
                \lIf{$N(u)^{u<} \cap Y = \emptyset \land N(u)^{<u} \cap Y = \emptyset$}{
                    $Del_1 \gets Del_1 \cup \set{u}$\label{algo:dsd:candd:del1:loop:smaller:end}
                } 
            }{\label{algo:dsd:candd:del1:loop:larger:start}
                \lIf{$N[u] \cap (X \setminus C) = \emptyset \land \size{N[u] \cap C } = 2$}{
                    $Del_2 \gets Del_2 \cup (N[u] \cap C)$
                }\label{algo:dsd:candd:del1:loop:larger:end}
            } \label{algo:dsd:candd:del1:loop:end}
        }
    \Return $C \setminus (Del_1 \cup Del_2)$\tcp*{$C$ is $\cand{X\setminus\set{v}}$}
    } 
    \Procedure{\DomList$(v, Y, X, C'\oplus \cand{Y}, \sig D')$}{
        
        \For{$u \in C' \oplus C(Y)$}{
            \For{$w \in N(u)^{u<}$}{
                \If{$u \notin D'_w(X)$}{
                    \lIf{$u \notin C'$}{
                        $D'_w \gets D'_w \cup \set{u}$ 
                    }
                    \lElse{
                        $D'_w \gets D'_w \setminus \set{u}$
                    }
                }
            }
        }
        
        \For{$u \in N(v)^{v<}$}{ \label{algo:dsd:domlist:neighbor:add:loop:start}
            \lIf{$u\in X$}{
                $D'_v \gets D'_v \cup \set{u}$
            }
        } \label{algo:dsd:domlist:neighbor:add:loop:end}

        \Return $\sig D'$\tcp*{$\sig D'$ is $\sig D(Y)$}
    }
\end{algorithm}

%% file: girth.tex
\section{Efficient Enumeration for Graphs with Girth at Least Nine}
\label{subsec:girth}

\input{enumdsg_short}

In this section, we propose an optimum enumeration algorithm \EnumDSG for graphs with girth at least nine, 
where the girth of a graph is the length of a shortest cycle in the graph.
That is, the proposed algorithm runs in constant amortized time per solution for such graphs. 
The algorithm is shown in Algorithm~\ref{algo:dsG}. 
To achieve constant amortized time enumeration, 
    we focus on the \name{local structure} $G_v(X)$ for $(X, v)$ of $G$ defined as follows:   
        $G_v(X) = G[(V \setminus N[X \setminus \cand{X}^{\le v}]) \cup \cand{X}^{\le v}]$. 
Fig.~\ref{fig:girth} shows an example of $G_v(X)$. 
$G_v(X)$ is a subgraph of $G$ induced by vertices that 
    (1) are dominated by vertices only in $\cand{X}^{\le v}$ or 
    (2) are in $\cand{X}^{\le v}$.  
Intuitively speaking, 
    we can efficiently enumerate solutions by using the local structure and ignoring vertices in $G \setminus G_v(X)$ 
    since 
    the number of solutions that are generated according to the structure is enough to reduce the \emph{amortized} time complexity to constant. 
We denote by $G(X) = G[(V \setminus N[X \setminus \cand{X}]) \cup \cand{X}]$ 
    the local structure for $(X, v_*)$ of $G$, 
    where $v_*$ is the largest vertex in $G$.

We first consider the correctness of \EnumDSG. 
The parent-child relation between solutions used in \EnumDSG is the same as in \EnumDS. 
Suppose that $X$ and $Y$ are dominating sets such that $X$ is the parent of $Y$. 
Recall that, from Lemma~\ref{lem:del}, 
$\cand{X} \setminus \cand{Y} = \Del{X}{v}$, 
    where $X = Y \cup \set{v}$. 
We denote by $f_v(u, X) = \True$ if there exists a neighbor $w$ of $u$ such that $w \in X \setminus \cand{X}^{\le v}$; 
Otherwise $f_v(u, X) = \False$. 
Thus, \CandG correctly computes $\Del[1]{X}{v}$ and $\Del[2]{X}{v}$
from line~\ref{alg:candg:del:start} to \ref{alg:candg:del:end}. 
Moreover, in line~\ref{algo:dsg:set:C}, 
vertices in $\Del[3]{X}{v}$ are removed from $\cand{X}$ and hence, 
\CandG also correctly computes $\cand{X \setminus \set{v}}$. 
Moreover, for each vertex $w$ removed from $G$ during enumeration, 
    $w$ is dominated by some vertices in $G$. 
Hence, by the same discussion as Theorem~\ref{theo:enum}, 
    we can show that \EnumDSG enumerates all dominating sets. 
In the remaining of this section, 
    we show the time complexity of \EnumDSG. 
Note that $G_v(X)$ does not include any vertex in $N[\Del[3]{X}{v}\setminus\set{v}]\setminus\cand{X}^{\le v}$.  
Hence, 
we will consider only vertices in $\Del[1]{X}{v} \cup \Del[2]{X}{v}\cup\set{v}$. 
We denote by $\DelG{X}{v} = \Del[1]{X}{v}\cup \Del[2]{X}{v}\cup\set{v}$.  
We first show the time complexity for updating the candidate sets. 

In what follows, 
if $v$ is the largest vertex in $\cand{X}$, then we simply write $f(u, X)$ as $f_v(u, X)$. 
We denote by $N'_v(u) = N_{G_v(X)}(u)$, $N'_v[u] = N'_v(u) \cup \set{u}$, and $d'_v(u) = \size{N'_v(u)}$ if no confusion arises. 
Suppose that $G$ and $G_v(X)$ are stored in an adjacency list, and 
neighbors of a vertex are stored in a doubly linked list and sorted in the ordering. 

\begin{lemma}
\label{lem:cand:dsg:del1}
    Let $X$ be a dominating set, $v$ be a vertex in $\cand{X}$, and $u$ be a vertex in $G$. 
    Then, $u \in \Del[1]{X}{v}$ if and only if 
    $N'_v[u] \cap X = \set{u, v}$ and $f_v(u, X) = \False$.
\end{lemma}
\begin{proof}
    The only if part is obvious 
    since $u, v \in \cand{X}^{\le v}$ and $N[u]\cap X = \set{u, v}$. 
    We next prove the if part. 
    Since $f_v(u, x) = \False$, 
    $N[u] \cap (X\setminus\cand{X}^{\le v}) = \emptyset$. 
    Moreover, 
    since $(N'_v[u] \cap X) \subseteq \cand{X}^{\le v}$, 
    $N[u] \cap X = N'_v[u] \cup (N[u] \cap (X\setminus\cand{X}^{v<})) = \set{u, v}$. 
    Hence, the statement holds. 
\end{proof}

\begin{lemma}
\label{lem:cand:dsg:del2}
    Let $X$ be a dominating set, $v$ be a vertex in $\cand{X}$, and $u$ be a vertex in $G$. 
    Then, $u \in \Del[2]{X}{v}$ 
    if and only if 
    there is a vertex $w$ in $G_v(X)$ such that $N'_v[w] \cap X = \set{u, v}$. 
\end{lemma}
\begin{proof}
    The only if part is obvious 
    since $u, v \in \cand{X}^{\le v}$ and  
    there is a vertex $w$ such that $N[w]\cap X = \set{u, v}$. 
    We next show the if part. 
    Since $w\in G_v(X)$,
    $w\in\cand{X}^{\le v}$ or $w \notin X \cup N[X\setminus\cand{X}^{\le v}]$. 
    Moreover, since $N'[w] = \set{u, v}$, $w\notin X$, that is, $w \notin \cand{X}$. 
    Hence, $w\notin N[X\setminus\cand{X}^{\le v}]$. 
    Therefore, $N[w] \cap X = (N'_v[w] \cap X) \cup (N[w]\cap(X\setminus\cand{X}^{v<})) = \set{u, v}$  
    and the statement holds. 
\end{proof}

\begin{lemma}
\label{lem:enumdsg:comp:cand}
    Let $X$ be a dominating set and $v$ be a vertex in $\cand{X}$. 
    Suppose that
    for any vertex $u$, 
    we can check the number of $u$'s neighbors in the local structure $G_v(X)$ and 
    the value of $f_v(u, X)$ in constant time. 
    Then, we can compute $\cand{X\setminus\set{v}}$ from $\cand{X}^{\le v}$ in $\order{d'_v(v)}$ time 
\end{lemma}

\begin{proof}
    Since $\Del[3]{X}{v} \cap \cand{X\setminus\set{v}} = \emptyset$, 
        $\cand{X\setminus\set{v}} \subseteq \cand{X}^{\le v}$. 
        Thus, we do not need to remove vertices in $\Del[3]{X}{v}$ from $\cand{X}^{\le v}$. 
    From Lemma~\ref{lem:cand:dsg:del1}, 
        for each vertex $u \in N'_v(v)$,
        we can check whether $u \in \Del[1]{X}{v}$ or not in constant time 
        by confirming that $f_v(u, X) = \False$ and $\size{N'_v(u)} = 2$.  
    Moreover, from Lemma~\ref{lem:cand:dsg:del2}, 
        for each vertex $w \in N'_v(v)$, 
        we can compute $\Del[2]{X}{v}$ 
        by listing vertices in $u \in \cand{X}^{\le v}$ such that $N'[w] \cap X = \set{u, v}$ or not. 
    Note that since any vertex in $X^{<v}$ belongs to $X$, 
        $N'[w] \cap X =\set{u, v}$ 
        if $f_v(w, X) = \False$, $\size{N'[w]} = 2$,
        and $u$ and $v$ are adjacent to $w$. 
    Hence, the statement holds. 
\end{proof}

\begin{lemma}
\label{lem:enumdsg:comp:local}
    Let $X$ be a dominating set, $v$ be a vertex in $\cand{X}$, and $Y = X\setminus\set{v}$. 
    Then, we can compute $G(Y)$ from $G_v(X)$ 
    in $\order{\sum_{u \in \DelG{X}{v}}d'_v(u) + \sum_{u \in G_v(X) \setminus G(Y)}d'_v(u)}$ time. 
    Note that $N'_v(u) = N_{G_v(X)}(u)$ and $d'_v(u) = \size{N'_v(u)}$. 
\end{lemma}

\begin{proof}
    From the definition, $V(G(Y)) \subseteq V(G_v(X))$. 
    Let us denote by $u$ a vertex in $G_v(X)$ but not in $G(Y)$ such that $u \neq v$. 
    This implies that 
        (A) $u$ is dominated by some vertex in $Y \setminus \cand{Y}$ and 
        (B) $u \notin \cand{Y}$. 
    Thus, for any vertex $u' \notin N'_v[\DelG{X}{v}\setminus\set{v}]$, $u' \in G_v(X)$ if and only if $u' \in G(Y)$. 
    Hence, 
        we can find such vertex $u$ by 
        checking whether for each vertex $w \in N'_v[\DelG{X}{v}]$, $w$ satisfies (A) and (B). 
    Before checking, 
        we first update the value of $f$. 
    This can be done by checking all the vertices in $N'_v[\DelG{X}{v}]$ and in $\order{1}$ time per vertex. 
    Hence, this update needs  $\order{\sum_{w \in \DelG{X}{v}}d'_v(w)}$ time. 
    If $w$ satisfies these conditions, that is, $f_v(w, X) = \False$, $f(w, Y) = \True$, and (B), 
        then we remove $w$ and edges that are incident to $w$ from $G_v(X)$. 
    This needs $\order{\sum_{w \in G_v(X) \setminus G(Y)}d'_v(w)}$ total time for removing vertices. 
    Thus, the statement holds. 
\end{proof}

From  Lemma~\ref{lem:enumdsg:comp:cand} and Lemma~\ref{lem:enumdsg:comp:local}, 
    we can compute the local structure and the candidate set of $Y$ from those of $X$ 
     in $\order{\sum_{u \in \DelG{X}{v}}d'_v(u) + \sum_{u \in G_v(X) \setminus G(Y)}d'_v(u)}$ time. 
We next consider the time complexity of the loop in line~\ref{algo:dsg:subloop}. 
In this loop procedure, \EnumDSG deletes all the neighbors $u$ of $v$ from $G_v(X)$
    if $u \notin \cand{X}^{\le v}$ 
    because for each descendant $W$ of dominating set $Y'$, 
    $v \in W \setminus \cand{W}$,  
    where $Y'$ is a child of $X$ and is generated after $Y$. 
Thus, 
this needs $\order{d'_v(v) + \sum_{u \in N'(v)\setminus X}d'_v(u)}$ time. 
Hence, 
from the above discussion, 
we can obtain the following lemma: 

\begin{lemma}
\label{lem:dsg:delte}
    Let $X$ be a dominating set, $v$ be a vertex in $\cand{X}$, and $Y = X \setminus \set{v}$. 
    Then, \AllChildren other than a recursive call runs in the following time bound: 
    \begin{equation}
    \label{eq:upper}
    \order{
        \sum_{u \in \DelG{X}{v}}d'_v(u) + 
        \sum_{u \in G_v(X) \setminus G(Y)}d'_v(u) + 
        \sum_{u\in N'_v(v)\setminus X}d'_v(u)
    }. 
    \end{equation}
\end{lemma}

Before we analyze the number of descendants of $X$, 
we show the following lemmas. 

\begin{lemma}
\label{lem:pendant}
    Let us denote by $Pen_v(X) = \inset{u \in \DelG{X}{v}}{d'_v(u) = 1}$. 
    Then, $\sum_{v\in\cand{X}}\size{Pen_v(X)}$ is at most $\size{\cand{X}}$. 
\end{lemma}

\begin{proof}
    Let $u$ be the largest vertex in $\cand{X}^{<v}$ and  
    $w$ be a vertex in $G_v(X) \cap \DelG{X}{v}$. 
    If $w\in\Del[1]{X}{v}$, 
    then $d'_u(w) = 0$ since $w \in N'_v(v)$. 
    Otherwise, $w\in\Del[2]{X}{v}$, then 
    $d'_u(w) = 0$ since a vertex $x$ such that $N'_v[x] = \set{w, v}$ is removed from $G_v(X)$. 
    Hence, $Pen_v(X) \cap Pen_u(X) = \emptyset$. 
    Moreover, for each $v\in\cand{X}$, $Pen_v(X)$ is a subset of $\cand{X}$. 
    Hence, the union of $Pen_v(X)$ is a subset of $\cand{X}$ for each $v\in\cand{X}$. 
\end{proof}

Let $v$ be a vertex in $\cand{X}$ and a pendant in $G_v(X)$.     
Since the number of such pendants is at most $\size{\cand{X}}$, 
the sum of degree of such pendants 
is at most $\size{\cand{X}}$ in each execution of \AllChildren without recursive calls. 
Hence, the cost of deleting such pendants is $\order{\size{\cand{X}}}$ time. 
Next, we consider the number of descendants of $X$. 
From Lemma~\ref{lem:pendant}, we can ignore such pendant vertices. 
Hence, for each $u \in \DelG{X}{v}$, 
we will assume that $d'_v(u) \ge 2$  below. 

\begin{figure}[t]
    \centering  
    \includegraphics[width=0.8\textwidth]{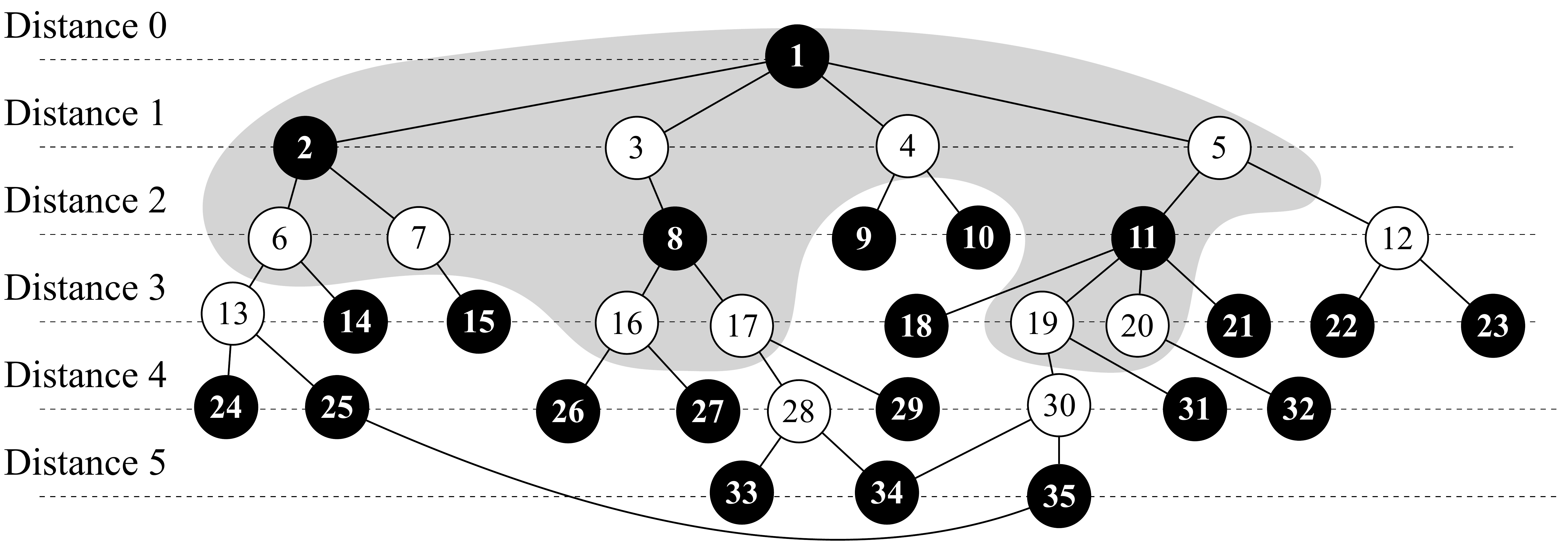}
    \caption{An example of $G_v(X)$, where $v = 1$. The vertices in the grey area are $\DelG{X}{v} \cup (G_v(X)\setminus G(Y)) \cup (N'_v(v)\setminus X)$.  
    Each horizontal line represents the distance between $1$ and any vertex. }
    \label{fig:girth}
\end{figure}

\begin{lemma}
\label{lem:cand}
    Let $X$ be a dominating set, $v$ be a vertex in $\cand{X}$, and $u$ be a vertex in $G_v(X)$.
    Then, $\size{N'_v[u] \cap \cand{X}^{\le v}} \ge 2$ if $u \notin \cand{X}$. 
    Otherwise, $\size{N'_v[u] \cap \cand{X}^{\le v}} \ge 1$. 
\end{lemma}

\begin{proof}
    If $u \in\cand{X}$, 
        then $u \in N'[u]\cap \cand{X}$. 
    We assume that $u \notin \cand{X}$.
    Thus, 
    $N'[u] \cap (X \setminus \cand{X}) = \emptyset$ from the definition of $G(X)$. 
    If $\size{N'[u] \cap \cand{X}} = 0$, then $u$ is not dominated by any vertex. 
    This contradicts $X$ is dominating set. 
    If $\size{N'[u] \cap \cand{X}} = 1$, 
        then $u$ is dominated only by the neighbor $w$ of $u$ in $\cand{X}$. 
    This contradicts $w \in \cand{X}$. 
    Hence, $\size{N[v] \cap \cand{X}} \ge 2$ if $v \notin\cand{X}$. 
\end{proof}

\begin{lemma}
\label{lem:child2:v}
    Let $X$ be a dominating set, $v$ be a vertex in $\cand{X}$, and $Y$ be a dominating set $X\setminus\set{v}$. 
    Then, $\size{\cand{Y}}$ is at least $\size{(N'_v(v)\cap X)\setminus\DelG{X}{v}}$. 
\end{lemma}

\begin{proof}
    Let $u$ be a vertex in $(N'_v(v)\cap X)\setminus\DelG{X}{v}$. 
    If $u \in \cand{X}$, 
    then $u$ is also a candidate vertex in $\cand{Y}$ since $u \notin \DelG{X}{v}$. 
    Suppose that $u \notin \cand{X}$. 
    Since $u \in G_v(X)$, 
    $u$ is dominated by only candidate vertices of $X$.
    However, since $u \in X$, $u$ dominates it self and thus, this contradicts. 
    Hence, the statement holds. 
\end{proof}

\begin{lemma}
\label{lem:child2:neighbor:v}
    Let $X$ be a dominating set, $v$ be a vertex in $\cand{X}$, 
    and $Y$ be a dominating set $X\setminus\set{v}$. 
    Then, $\size{\cand{Y}}$ is at least
    $\sum_{u\in N'_v(v)\setminus X}(d'_v(u) - 1)$. 
\end{lemma}

\begin{proof}
    Let $u$ be a vertex in $N'_v(v) \setminus X$.  
    That is, $u \notin \cand{X}$ and $N'_v(u) \subseteq \cand{X}$. 
    Thus,  from Lemma~\ref{lem:cand}, 
        there is a vertex $w\in N'_v(u)$ such that $w < v$.
    We consider the following two cases: 
    (A) If $N'_v(u) = \set{v, w}$, 
        then $w \in \DelG{X}{v}$. 
    From the assumption, 
        $w$ has at least one neighbor $x$ such that $x \neq u$. 
    If $x \notin \cand{X}$, 
        then there is a neighbor $y \in \cand{X}$ such that $y \neq w$. 
    Suppose that $y \in \DelG{X}{v}$. 
    This implies that there is a cycle with length at most six. 
    This contradicts the girth of $G$. 
    Hence, $y \notin \DelG{X}{v}$ and 
        $Y\setminus\set{y}$ is a dominating set. 
    If $x \in \cand{X}$, 
        then $x \notin \DelG{X}{v}$ from the definition of $\DelG{X}{v}$ and the girth of $G$. 
    Hence, 
        $Y\setminus\set{x}$ is a dominating set. 
    (B) Suppose $N'_v(u)$ has a vertex $z \in \cand{X}$ such that $z \neq v$ and $z \neq w$. 
    If both $z$ and $w$ are in $\DelG{X}{v}$, 
        then from the definition of $\DelG{X}{v}$ and the girth of $G$, 
            $G$ has a cycle with length at most five.
    Thus, without loss of generality, 
        we can assume that $z \notin \DelG{X}{v}$. 
    This allows us to generate a child $Y \setminus \set{z}$ of $Y$. 
    Since the girth of $G$ is at least nine, 
        all children of $Y$ generated above are mutually distinct. 
    Hence, the statement holds. 
\end{proof}

\begin{lemma}
\label{lem:child2:del}
    Let $X$ be a dominating set, 
        $v$ be a vertex in $\cand{X}$, 
        and $Y$ be a dominating set $X\setminus\set{v}$. 
    Then, $\size{\cand{Y}}$ is at least
    $\sum_{u \in \DelG{X}{v}\setminus\set{v}} \left(d'_v(u) - 1\right)$. 
\end{lemma}

\begin{proof}
    Let $u$ be a vertex in $\DelG{X}{v}\setminus\set{v}$. 
    From the assumption, there is a neighbor $w$ of $u$ in $G(X)$. 
    We consider the following two cases: 
    (A) Suppose that $w$ is in $G(Y)$.
    Since $u$ is in $Y \setminus \cand{Y}$,  $w \in \cand{Y}$. 
    Hence, $Y\setminus\set{w}$ is a child of $Y$. 
    Suppose that for any two distinct vertices $x, y$ in $\DelG{X}{v}\setminus\set{v}$, 
        they have a common neighbor $w'$ in $G(Y)$. 
    If both $x$ and $y$ are in $\Del[2]{X}{v}$, 
        then there exist two vertex $z_x, z_y$ such that $N'_v[z_x]\cap X = \set{x, v}$ and 
        $N'_v[z_y] \cap X = \set{y, v}$, respectively. 
    Therefore, 
        there is a cycle $(v, z_x, x, w', y, z_y, v)$ with length six.
    As with the above, 
        if $x$ or $y$ in $\Del[1]{X}{v}$, 
            then there exists a cycle with length less than six since $\set{x, v} \in G$ or $\set{x, v} \in G$. 
    This contradicts of the assumption of the girth of $G$. 
    Hence, any pair vertices in $\DelG{X}{v}$ has no common neighbors.
    Thus, in this case, all grandchildren of $X$ are mutually distinct. 
    (B) Suppose that $w$ is not in $G(Y)$. 
    Thus, if $w \in \cand{X}$, 
        then $w \in \DelG{X}{v}$. 
    This implies that there is a cycle including $w$ and $u$ whose length is less than six. 
    Hence,  $w$ is not in $\cand{X}$. 
    Then, from Lemma~\ref{lem:cand}, 
        there is a vertex $z$ in $N'_v(w) \cap \cand{X}$ such that $z \neq u$. 
    Since $u \in \DelG{X}{v}\setminus\set{v}$, 
        there is an edge between $u$ and $v$, or
        there is a vertex $c$ such that $\set{u, c}$ and $\set{v, c}$ are in $G_v(X)$. 
    Again, if $z$ is in $\DelG{X}{v}$, 
        then there is a cycle with length less than seven. 
    Thus, $z$ still belongs to $\cand{Y}$ and $X\setminus\set{v, z}$ is a dominating set. 
    Next, we consider the uniqueness of $X \setminus\set{v, z}$. 
    If there is a vertex $w'$ such that $w' \in N'_v(u)$, $w' \neq w$, $w$ and $w'$ share a common neighbor $u'$ other than $u$, 
    then $(u, w, u', w')$ is a cycle. 
    Hence, any pair neighbors of $u$ has no common neighbors.
    As with the above, any two distinct vertices in $\DelG{X}{v}\setminus\set{v}$ also has no common vertex like $z$. 
    If there are two distinct vertex $u, u'\in \DelG{X}{v}$ such that $u$ and $u'$ has a common vertex like $z$, 
    then there is a cycle with length at most eight even if $u, u'\in\Del[2]{X}{v}$.
    This contradicts the assumption of the girth, and thus, the statement holds. 
\end{proof}

\begin{lemma}
\label{lem:child3}
    Let $X$ be a dominating set $v$ be a vertex in $\cand{X}$, and $Y$ be a dominating set $X\setminus\set{v}$. 
    Then, the number of children and grandchildren of $Y$ is at least 
    $\sum_{u \in 
            G_v(X)\setminus
            \left(G(Y)\cup\DelG{X}{v}\cup N'_v(v) \right)
        } 
        \left(d'_v(u) - 1\right)$. 
\end{lemma}

\begin{proof}
    Let $u$ be a vertex in $G_v(X) \setminus (G(Y)\cup\DelG{X}{v}\cup N'_v(v))$.
    Since $u \notin \DelG{X}{v}$ and $u \in G_v(X)\setminus G(Y)$, 
        $u$ is not in $X$.
    Since $\size{N'_v(u)\cap \cand{X}^{\le v}}$ is greater than or equal to two from Lemma~\ref{lem:cand}, 
        there are two distinct vertices $w, w'$ in $N'_v(u)$. 
    We assume that $w, w' \in\DelG{X}{v}$. 
    From Lemma~\ref{lem:del}, the distance between $w$ and $v$ is at most two. 
    Similarly, the distance between $w'$ and $v$ is at most two. 
    Hence, there is a cycle with the length at most six since $w \neq v$ and $w' \neq v$. 
    Thus, without loss of generality, 
        we can assume that $w \notin\DelG{X}{v}$.     
    (A) Suppose that $\size{N'_v(u)} = 2$. 
    If there is a vertex $u' \in G_v(X)\setminus(G(Y) \cup \DelG{X}{v}\cup N'_v(v))$ 
        such that $u' \neq u$ and $w \in N'(u)$, 
        then as with Lemma~\ref{lem:child2:neighbor:v}, 
            there is a short cycle. 
    Hence, 
        for each vertex such as $u$, 
            there is a corresponding dominating set $X \setminus \set{v, w}$. 
    (B) Suppose that there is a neighbor $w'' \in N'_v(u) \cap \cand{X}$. 
    Then, as mentioned in above, 
        there is a dominating set $X\setminus\set{v, w, w''}$. 
    In addition, 
        by the same discussion as Lemma~\ref{lem:child2:del}, 
            such generated dominating sets are mutually distinct. 
    (C)  Suppose that there is a neighbor $w'' \in N'_v(u) \setminus \cand{X}$. 
    From Lemma~\ref{lem:cand}, 
        there are two vertices $z, z' \in N'(w'')\cap\cand{X}$. 
    Then,  $z \notin\DelG{X}{v}$ or $z' \notin\DelG{X}{v}$, 
            and thus, we can assume that $z \notin\DelG{X}{v}$. 
    Therefore, there is a dominating set $X\setminus\set{v, w, z}$. 
    Next, we consider the uniqueness of grandchildren of $Y$. 
    Moreover, 
    if there is a vertex $u'$ such that $w, y \in N'(u')$ holds, 
        such that $z\in N'(y)$. 
    Then, there is a cycle $(u, w, u', y, z, w'')$ with the length six. 
    Hence, grandchildren of $Y$ are mutually distinct for each $u \in G(X)\setminus G(Y)\setminus\DelG{X}{v}$. 
    Thus, from (A), (B), and (C), the statement holds.
\end{proof}

Note that for any pair of candidate vertices $v$ and $v'$, 
$X\setminus\set{v}$ and $X\setminus\set{v'}$ do not share their descendants.
Thus, 
from Lemma~\ref{lem:child2:v}, Lemma~\ref{lem:child2:neighbor:v}, Lemma~\ref{lem:child2:del}, and Lemma~\ref{lem:child3}, 
we can obtain the following lemma: 

\begin{lemma}
\label{lem:child}
    Let $X$ be a dominating set.
    Then, the sum of the number of $X$'s children, grandchildren, and great-grandchildren 
    is bounded by the following order:  
    \begin{equation}
    \label{eq:lower}
     \bigomega{
     \size{\cand{X}} +
     \sum_{v\in\cand{X}}\left(
            \sum_{u \in \DelG{X}{v}}d'_v(u) + 
            \sum_{u \in G_v(X) \setminus G(Y)}d'_v(u) + 
            \sum_{u\in N'_v(v)\setminus X}d'_v(u)
        \right)}. 
    \end{equation}
\end{lemma}
    
From Lemma~\ref{lem:dsg:delte}, Lemma~\ref{lem:pendant}, and Lemma~\ref{lem:child}, 
each iteration outputs a solution in constant amortized time. 
Hence, by the same discussion of Theorem~\ref{theo:dsd:time},  
we can obtain the following theorem. 

\begin{theorem}
\label{theo:girth}
    For an input graph with girth at least nine, 
    \EnumDSG enumerates all dominating sets 
    in $\order{1}$ time per solution by using $\order{n + m}$ space. 
\end{theorem}
\begin{proof}
    The correctness of \EnumDSG is shown by Theorem~\ref{theo:enum}, Lemma~\ref{lem:cand:dsg:del1}, and Lemma~\ref{lem:cand:dsg:del2}.
    By the same discussion with Theorem~\ref{theo:dsd:time}, 
        the space complexity of \EnumDSG is $\order{n + m}$. 
    We next consider the time complexity of \EnumDSG. 
    From Lemma~\ref{lem:dsg:delte}, Lemma~\ref{lem:pendant}, and Lemma~\ref{lem:child}.
    we can amortize the cost of each recursion 
        by distributing $\order{1}$ time cost to the corresponding descendant discussed in the above lemmas.
    Thus, the amortized time complexity of each recursion becomes $\order{1}$. 
    Moreover, each recursion outputs a solution. 
    Hence,  \EnumDSG enumerates all solutions in $\order{1}$ amortized time per solution. 
\end{proof}


%% file: enumdsg_short.tex
\begin{algorithm}[t]
  \caption{\texttt{EDS-G} enumerates all dominating sets in $\order{1}$ time per solution for a graph with girth at least nine. }
  \label{algo:dsG}
    \Procedure(\tcp*[f]{$G$: an input graph}){\EnumDSG{$G = (V, E)$}}{
        \lFor{$v \in V$}{
            $f_v \gets \False$
        }
        \AllChildren$(V, V, \set{f_1, \dots, f_{\size{V}}}, G)$\;
    }
  \Procedure{\AllChildren$(X, C, F, G)$}{
        Output $X$\;
        \For(\tcp*[f]{$v$ is the largest vertex in $C$}){$v \in \cand{X}$}{\label{algo:dsg:mainloop}
            $Y \gets X\setminus\set{v}$\;
            $(\cand{Y}, F(Y), G(Y)) \gets$ \CandG$(v, C, F, G)$\;
            \AllChildren$(Y, \cand{Y}, F(Y), G(Y))$\;
            \For{$u \in N_G(v)$}{\label{algo:dsg:subloop}
                \lIf{$u \in C$}{
                    $f_u \gets \True$
                }\lElse{
                    $G \gets G \setminus \set{u}$
                }
            }
            $G \gets G \setminus \set{v}$\;\label{algo:dsg:remove:v:from:gv}
            $C \gets C \setminus \set{v}$\tcp*{Remove vertices in $\Del[3]{X}{v}$. }\label{algo:dsg:set:C}
        }
    }

    \Procedure{\CandG$(v, C, F, G)$}{
        $Del_1 \gets \emptyset$; $Del_2 \gets \emptyset$\;
        \For{$u \in N_G(v)$}{\label{alg:candg:del:start}
            \lIf{$N_G[u] \cap X = \set{u, v}$ and $f_u = \False$}{
                $Del_1 \gets Del_1 \cup \set{u}$
            }\lElseIf{$\exists w(N_G[u] \cap X = \set{w, v})$}{
                $Del_2 \gets Del_2 \cup \set{w}$
            }
        }\label{alg:candg:del:end}
        $C' \gets C \setminus (Del_1 \cup Del_2 \cup \set{v})$\;
        \For(\tcp*[f]{Lemma~\ref{lem:enumdsg:comp:local}}){$u \in N'[Del_1 \cup Del_2]$}{
            $f_u \gets \True$\;
            \lIf{$u \notin C'$}{
               $G \gets G \setminus \set{u}$
            }
        }
        \lIf{$f_v = \True$}{$G \gets G \setminus \set{v}$}
        \Return $(C', F, G)$\;
    }
\end{algorithm}

%% file: conc.tex
\section{Conclusion}
\label{sec:conc}
In this paper, we proposed two enumeration algorithms.
\EnumDSD solves the dominating set enumeration problem 
in $\order{k}$ time per solution by using $\order{n + m}$ space,
where $k$ is a degeneracy of an input graph $G$. 
Moreover, \EnumDSG solves this problem in constant time per solution
if an input graph has girth at least nine.

Our future work includes to develop efficient dominating set enumeration algorithms for dense graphs. 
If a graph is dense, 
then $k$ is large and $G$ has many dominating sets.
For example, in the case of complete graphs, 
$k$ is equal to $n - 1$ and every nonempty subset of $V$ is a dominating set.
That is, the number of solutions for a dense graph is much larger than 
that for a sparse graph. 
This allows us to spend more time in each recursive call. 
However, \EnumDSD is not efficient for dense graphs 
although the number of solutions is large. 
Moreover, if $G$ is small girth, that is, $G$ is dense then 
\EnumDSG does not achieve constant amortized time enumeration. 
Hence,  
the dominating set enumeration problem for dense graphs is interesting.